\documentclass[a4paper,11pt]{article}

\usepackage{latexsym,graphicx,fullpage}
\usepackage{amsthm,amsmath,amssymb,enumerate}
\usepackage{bm}
\usepackage{ifpdf}
\usepackage{color}
\usepackage{algorithm}
\usepackage[noend]{algpseudocode}
\usepackage{nicefrac}
\usepackage{wrapfig}

\allowdisplaybreaks

\definecolor{Darkblue}{rgb}{0,0,0.4}
\definecolor{Brown}{cmyk}{0,0.81,1.,0.60}
\definecolor{Purple}{cmyk}{0.45,0.86,0,0}
\newcommand{\mydriver}{hypertex}
\ifpdf
 \renewcommand{\mydriver}{pdftex}
\fi
\usepackage[breaklinks,\mydriver]{hyperref}
\newcommand{\lref}[2][]{\hyperref[#2]{#1~\ref*{#2}}}

\newtheorem{theorem}{Theorem}[section]
\newtheorem{definition}[theorem]{Definition}

\newtheorem{lemma}[theorem]{Lemma}

\newtheorem{claim}[theorem]{Claim}

\newtheorem{corollary}[theorem]{Corollary}

\numberwithin{algorithm}{section}

\newenvironment{subproof}[1][\proofname]{%
  \begin{proof}[#1]%
}{%
  \end{proof}%
}

\newcommand{\junk}[1]{}
\newcommand{\ignore}[1]{}

\newcommand{\R}[0]{{\ensuremath{\mathbb{R}}}}

\newcommand{\eps}{\varepsilon}

\newcommand{\E}{\mathbb{E}}

\newcounter{note}[section]

\newcommand{\qedsymb}{\hfill{\rule{2mm}{2mm}}}

\newcommand{\initOneLiners}{%
    \setlength{\itemsep}{0pt}
    \setlength{\parsep }{0pt}
    \setlength{\topsep }{0pt}
}

\newcommand{\squishlist}{
 \begin{list}{$\bullet$}
  { \setlength{\itemsep}{0pt}
     \setlength{\parsep}{3pt}
     \setlength{\topsep}{3pt}
     \setlength{\partopsep}{0pt}
     \setlength{\leftmargin}{1.5em}
     \setlength{\labelwidth}{1em}
     \setlength{\labelsep}{0.5em} } }

\newcommand{\squishend}{
  \end{list}  }

\begin{document}

\title{{\bf On a generalization of iterated and randomized rounding}
  \thanks{CWI and TU Eindhoven, Netherlands. Email: {\tt bansal@gmail.com}. Supported by the ERC Consolidator Grant 617951
  and the NWO VICI grant 639.023.812.
	This work was done in part while the author was visiting the \emph{Bridging
      Discrete and Continuous Optimization} program at the Simons Institute for the Theory of Computing.}
}
\author{Nikhil Bansal}
\date{}
\maketitle

 \begin{abstract}
We give a general method for rounding linear programs that combines the commonly used iterated rounding and randomized rounding techniques. In particular, we show that whenever iterated rounding can be applied to a problem with some {\em slack}, there is a randomized procedure that returns an integral solution that satisfies the guarantees of iterated rounding and also has concentration properties. 
We use this to give new results for several classic problems such as rounding column-sparse LPs, makespan minimization on unrelated machines, degree-bounded spanning trees and multi-budgeted matchings.
\end{abstract}



\section{Introduction}
A powerful approach in approximation algorithms is to formulate the 
problem at hand as a $0$-$1$ integer program
and consider some efficiently solvable relaxation for it. Then, given
some fractional solution $x\in [0,1]^n$ to this relaxation, apply a suitable
{\em rounding} procedure to $x$ to obtain an integral $0$-$1$ solution.
Arguably the two most basic and extensively studied techniques for rounding such relaxations are randomized rounding and iterated rounding.

{\bf Randomized Rounding.} Here, the fractional values $x_i \in [0,1]$ are
interpreted as probabilities, and used to round the variables independently
to $0$ or $1$.
A key property of this rounding is that each linear constraint
is preserved in expectation, and its value is tightly  concentrated around
its mean as given by Chernoff bounds, or more generally
Bernstein's inequality (definitions in Section \ref{s:prel}).
Randomized rounding is well-suited to problems where the  constraints do not have much structure, or when they are {\em soft} and some error can be tolerated. Sometimes these errors can be fixed by applying problem-specific alteration steps.
We refer to \cite{SW11, Vazirani01} for various applications of randomized rounding.

{\bf Iterated Rounding.} This on the other hand, is based on linear algebra and is useful for problems with {\em hard} combinatorial constraints or when the constraints have some interesting structure. Here, the rounding proceeds in several iterations $k=0,1,2,\ldots$, until
all variables are rounded to $0$ or $1$. Let $x^{(k)} \in \R^n$ denote the  
solution at the beginning
of iteration $k$, and let $n_k$ denote the number
of fractional variables in $x^{(k)}$  (i.e.~those strictly between $0$ and $1$). 
Then one (cleverly) chooses some collection of linear constraints on these $n_k$ fractional variables, say specified by rows of the 
matrix $W^{(k)}$, with dimension
$\dim(W^{(k)}) \leq n_k-1$, and updates the solution as $x^{(k+1)} = x^{(k)} + y^{(k)}$ by some
arbitrary non-zero vector $y^{(k)}$ satisfying $W^{(k)} y^{(k)}=\bf{0}$ so that some
fractional variable reaches $0$ or $1$.
The process is then iterated with $x^{(k+1)}$.
Note that once a variable reaches $0$ or $1$ it stays fixed.

Despite its simplicity, this method is extremely powerful and
most basic results in combinatorial optimization
such as the integrality of matroid, matroid-intersection and non-bipartite
matching polytopes follow very cleanly using this approach.
Similarly, several breakthrough results for problems such
as degree-bounded spanning trees, survivable network design
and rounding for column-sparse LPs
were obtained by this method. An excellent reference is \cite{Lau-book}.

\subsection{Our Results}
Motivated by several problems that we describe in Section \ref{s:appl}, 
a natural question is whether the strengths these of two seemingly different techniques can be combined to give a more powerful rounding approach.

Our main result is that such an algorithm exists, and we call it {\em sub-isotropic rounding}. In particular, it combines iterated and randomized rounding in a completely generic way and significantly the extends
the scope of previous dependent rounding techniques.
Before describing our result, we need some definitions.

Let $X_1,\ldots,X_n$ be  independent $0$-$1$ random variables with mean $x_i = \E[X_i]$ and $a_1,\ldots,a_n$ be arbitrary reals (possibly negative), then the sum
$S =\sum_i a_i X_i$ is concentrated about its mean and satisfies the following tail bound \cite{BLM13},
\begin{equation}
\label{eq:berns}
\text{(Bernstein's inequality)}  \quad \quad \Pr[ S - \E[S] \geq t] \leq \exp\left(-\frac{t^2}{2(\sum_i a_i^2(x_i-x_i^2) + Mt/3)}\right)
 \end{equation}
where $M = \max_{i} |a_i|$. The lower tail follows by applying the above to $-X_i$, and the standard Chernoff bounds correspond to \eqref{eq:berns} when $a_i\in[0,1]$ for $i\in [n]$ (details in Section \ref{s:prel}).

The following relaxation of Bernstein's inequality will be extremely relevant for us.
\begin{definition}[$\beta$-concentration] Let $\beta \geq 1$. For a vector valued random variable $X=(X_1,\ldots,X_n)$ where $X_i$ are possibly dependent $0$-$1$ random variables, we say that $X$ is $\beta$-concentrated around the mean $x = (x_1,\ldots,x_n)$ where $x_i = \E[X_i]$,
if for every $a \in \R^n$, $a \cdot X$ is well-concentrated and satisfies Bernstein's inequality  up to factor $\beta$ in the exponent, i.e.~
\begin{equation}
 \Pr[ a \cdot X- \E[a \cdot X] \geq t] \leq \exp\left(-\frac{t^2/\beta}{2(\sum_i a_i^2(x_i-x_i^2) + Mt/3)}\right)
 \end{equation}
\end{definition}

{\bf Main result.} We show that whenever iterated rounding can be applied to a problem so that in any iteration $k$, there is some {\em slack} in the sense that  
$\dim(W^{(k)}) \leq (1-\delta) n_k$ for some $\delta>0$,
then $O(1/\delta)$-concentration can be achieved for free.
More precisely, we show the following.

\begin{theorem}
\label{thm:main}
 Let $P$ be a problem for which there is an iterated
rounding algorithm $A$,
that at iteration $k$, chooses a subspace $W^{(k)}$ with 
$\dim(W^{(k)}) \leq  (1-\delta)n_k$, where $\delta>0$. 
Then there is a polynomial time randomized algorithm that given a starting solution $x \in [0,1]^n$, returns $X \in \{0,1\}^n$ such that
\begin{itemize}
\item With probability $1$, $X$ satisfies all the guarantees of
the iterated rounding algorithm $A$.
\item $\E[X_i]=x_i$ for every variable $i$, and $X$ is $\beta$-concentrated with $\beta = 20/\delta$.
\end{itemize}
\end{theorem}
A simple example in Appendix \ref{a:tight} shows that the dependence $\beta = \Omega(1/\delta)$ cannot be improved beyond constant factors. 

The generality of Theorem \ref{thm:main} directly gives new results for several problems where iterated rounding gives useful guarantees. All one needs to show is that the original iterated rounding argument for the problem
can be applied with some slack, which is often straightforward and only worsens the approximation guarantee slightly. Before describing these applications in Section \ref{s:appl}, we discuss some prior work on dependent rounding to place our results and techniques in the proper context.

\subsection{Comparison with Dependent Rounding Approaches}
Motivated by problems that involve both soft and hard constraints, there has been extensive work on developing {\em dependent rounding techniques}, that round the fractional solution in some correlated way to satisfy both the hard constraints and ensure some concentration properties.
Such problems arise naturally in many ways. E.g.~the hard constraints might arise from an underlying combinatorial object such as spanning tree or matching that needs to be produced, and the soft constraints may arise due to multiple budget constraints, or when the object to be output is used as input to another problem and needs to satisfy  additional properties, see e.g.~\cite{CVZ10,CVZ11,GKPS06, AGMOS10}.

Some examples of such methods include swap rounding \cite{CVZ10,CVZ11}, randomized pipage \cite{AFK,Srin01,GKPS06,HO14},  maximum-entropy sampling \cite{AGMOS10,SV14,AS10},
rounding via discrepancy \cite{LM12, R12, BN16} and gaussian random walks \cite{PSV17}.
A key idea here is that the weaker property of  
{\em negative} dependence (instead of independence) also suffices to get concentration. There is a rich and deep theory of
negative dependence and various notions such as negative correlation, negative cylinder dependence,
negative association, strongly rayleigh property and determinantal measures, imply interesting concentration properties \cite{Pem00, BBL09, DP09}.
This insight has been extremely useful and for many general problems such as those involving assignment or matroid polytopes, one can exploit the underlying combinatorial structure to design rounding approaches that ensure negative dependence between all or some  suitable collection of random variables.  

{\bf Limitations.} Even though very powerful and ingenious, these methods are also limited by the fact that requiring negative dependence can substantially restrict the kinds of rounding steps that can be designed, and the problems they can be applied to.
Moreover, even when such a rounding is possible, it typically requires a lot of creativity and careful understanding of the problem structure to come up with the rounding.

{\bf Our approach.} In contrast, Theorem \ref{thm:main} makes no assumption on the structure of the problem and by working with the more relaxed notion of  $\beta$-concentration, we can get rid of the need for negative dependence.
Moreover, our algorithm needs no major ingenuity to apply, 
and for most problems minor tweaks to previous iterated rounding algorithms suffice to create some {\em slack}.

\subsection{Motivating problems and Applications}
\label{s:appl}
We now give several applications and briefly discuss why they seem beyond the reach of current dependent rounding methods.

\subsubsection{Rounding Column-sparse LPs}
Let $x \in [0,1]^n$ be some fractional solution satisfying $Ax = b$, where $A \in \R^{m \times n}$ is an $m \times n$ matrix. 
The celebrated Beck-Fiala algorithm \cite{BF81} gives an integral solution $X$ so that $|AX - Ax|_\infty \leq t$, where $t$ is the maximum $\ell_1$ norm of the columns of $A$. This is substantially better than randomized rounding for small $t$.

Many problems however, involve both some column-sparse constraints that come from the underlying combinatorial problem, and some general arbitrary constraints which might not have much structure. This motivates the following natural question.

{\bf The problem.} Let $M$ be a linear system with two sets of constraints given by matrices $A$ and $B$, where $A$ is column-sparse, while $B$ is arbitrary. Given some fractional solution $x$, can we round it to
get error $O(t)$ for rows of $A$, while doing no worse than randomized rounding for $B$?

Note that simply applying iterated rounding on the rows of $A$ gives no control on the error for $B$. Similarly, just doing randomized rounding will not give $O(t)$ error for $A$. Also as $A$ and $B$ can be completely arbitrary, previous negative dependence based techniques do not seem to apply.


 {\bf Solution.} We show that a direct modification of the Beck-Fiala argument gives slack $\delta$, for any $\delta \in [0,1)$, while worsening the   error bound slightly to $t/(1-\delta)$. Setting, say $\delta=1/2$ and applying Theorem \ref{thm:main} gives $X \in \{-1,1\}^n$ that (i) has error at most $2t$ for rows of $A$, (ii) satisfies $\E[X_i]=x_i$ and is $O(1)$-concentrated, thus giving similar guarantees as randomized rounding for the rows of $B$.
In fact, the solution produced by the algorithm will satisfy concentration for all linear constraints and not just for the rows of $B$.

{\bf Koml\'{o}s Setting.} We also describe an extension to the so-called Koml\'{o}s setting, where the error depends on the maximum $\ell_2$ norm of columns of $A$.

These results are described in Section \ref{s:bf}.



\subsubsection{Makespan minimization on unrelated machines}
The classic makespan minimization problem on unrelated machines is the following. Given $n$ jobs and $m$ machines, where each job $j \in [n]$ has arbitrary size $p_{ij}$ on machine $i\in [m]$, assign the jobs to machines to minimize the maximum machine load. In a celebrated result, \cite{LST} gave a rounding method with additive error $p_{\max} := \max_{ij} p_{ij}$. 
In many practical problems however, there are other soft resource constraints and side constraints that are added to the  fractional formulation.
So it is useful to find a rounding that satisfies these approximately but still violates the main load constraint by only $O(p_{\max})$.
This motivates the following natural problem.

{\bf The Problem.}
Given a fractional assignment $x$, find an integral assignment $X$ with additive error $O(p_{\max})$ and that satisfies $\E[X_{ij}]=x_{ij}$ and
concentration for all linear functions of $x_{ij}$.

Questions related to finding a good assignment with some concentration properties have been studied before \cite{GKPS06,AFK,CVZ11}, and several methods such as randomized pipage and swap rounding have been developed for this.  
However, these methods crucially rely on the underlying matching structure and round the variables alternately along cycles, which limits them in various ways: either they give partial assignments, or only get concentration for edges incident to a vertex.

{\bf Solution.} We show that the iterated rounding proof of the result of \cite{LST} can be easily modified to  work for any slack $\delta \in [0,1/2)$ while giving additive error $p_{\max}/(1-2\delta)$. Theorem \ref{thm:main} (say, with $\delta=1/4$), thus gives a solution that has additive error at most 
$2p_{\max}$ and satisfies $O(1)$-concentration.

The result also extends naturally to the $k$ resource setting, where $p_{ij}$ is a $k$-dimensional vector. These results are described in Section \ref{s:makespan}

\subsubsection{Degree-bounded Spanning Trees and Thin Trees}
In the minimum cost degree-bounded spanning tree problem, we are 
given an undirected graph $G=(V,E)$ with edge costs $c_e \geq 0$ for $e \in E$, and integer degree bounds $b_v$ for $v \in V$, 
and the goal is to find a minimum cost spanning tree satisfying the degree bounds.
In a breakthrough result, Singh and Lau \cite{SL07} gave an iterated rounding algorithm that given any fractional spanning tree $x$, finds a spanning tree $T$ with cost at most $c \cdot x$ and degree violation $+1$. 

The celebrated thin-tree conjecture (details in Appendix \ref{a:thintree}) asks if given a fractional spanning tree $x$, there is a spanning tree $T$ satisfying $ \delta_T(S) \leq \beta \delta_x(S)$ for every $S\subset{V}$, where $\beta=O(1)$. Here $\delta_T(S)$ is the number of edges of $T$ crossing $S$, and $\delta_x(S)$ is the $x$-value crossing $S$.

 The result of \cite{SL07} implies that the degree $\delta_T(v) \leq 2 \delta_x(v)$ of every vertex $v$.
However, despite remarkable progress \cite{AO15}, the best known algorithmic results for the thin-tree problem give $\beta=O(\log n /\log \log n)$ \cite{AGMOS10,CVZ10,SV14,HO14}. 
The motivates the following natural question as a first step towards the thin-tree conjecture.

{\bf The Problem.} Can we find a spanning tree with $\beta=O(1)$ for single vertex cuts and $\beta=O(\log n/\log \log n)$ for general cuts?

The current algorithmic methods for thin-trees crucially rely on the negative dependence properties of spanning trees, which do not give anything better for single vertex cuts
(e.g.~even if $b_v=2$ for all $v$, by a balls and bins argument a random tree will have maximum degree $\Theta(\log n/\log \log n)$). On the other hand, if we only care about single vertex cuts the methods of \cite{SL07} do not  give anything for general cuts.

{\bf Solution.} We show that the iterated rounding algorithm of \cite{SL07} can be easily modified to create slack $\delta \in (0,1/2)$ while violating the degree bounds by at most $2/(1-2\delta)$. Applying Theorem \ref{thm:main} with $\delta = 1/6-\epsilon$ thus gives a distribution supported on trees with degree violation $+2$, and has $O(1)$ concentration. By  a standard cut counting argument \cite{AGMOS10}, the concentration property implies $O(\log n/\log n \log n)$-thinness for every cut.

We describe these results in Section \ref{s:spanning-tree}, where in fact we consider the more general 
setting of the minimum cost degree bounded matroid basis problem.

\subsubsection{Multi-budgeted bipartite matchings}
In the above examples, it was relatively easy to create slack 
since the number of hard combinatorial constraints were bounded a constant factor away from $n_k$,
and the slack could be introduced in the {\em soft} constraints (e.g.~machine load or vertex degrees) while worsening the approximation slightly. 

As a different type of illustrative example, we now consider the perfect matching problem in bipartite graphs. Here, and more generally in matroid intersection, one needs to maintain tight rank constraints for two matroids, which typically requires $n-1$ linearly independent constraints for $n$ elements, and it is not immediately clear how to introduce slack.

{\bf Problem.} Let $G=(V,E)$ be a bipartite graph with $V=L\cup R$ and $|L|=|R|$, and given a fractional perfect matching $x$ defined by $\sum_{e \in \delta(v)} x_e =1 $ for all $v\in V$, and $x_e \in [0,1]$ for all $e\in E$. Can we round it to a perfect or almost perfect matching while satisfying $O(1)$-concentration.

Building on the work of \cite{AFK}, \cite{CVZ11} designed a beautiful randomized swap rounding procedure that for any $\delta>0$,
finds an almost perfect matching where each vertex is matched with probability at least $1-\delta$, and satisfies $O(1/\delta)$-concentration. They also extend this result to non-bipartite matchings and matroid intersection. 

We give an alternate proof of this result using our framework.
Our proof is different from that in \cite{CVZ11} and is more in the spirit of iterated rounding where we carefully choose the set of constraints to maintain as the rounding proceeds. This is described in Section \ref{s:matching}.

\subsection{Overview of Techniques}
We now give a high level overview of our algorithm and  analysis.
The starting observation is that randomized rounding can be viewed as a iterative algorithm, by doing a standard Brownian motion on the cube as follows.
Given $x^{(0)}$ as the starting fractional solution, consider a random walk in the $[0,1]^n$ cube starting at $x^{(0)}$, with tiny step size $\pm \gamma$ chosen independently for each coordinate, where upon reaching a face of the cube (i.e.~some $x_i$ reaches $0$ or $1$) the walk stays on that face.
The process stops upon reaching some vertex $X=(X_1,\ldots,X_n)$ of the cube. By the martingale property of random walks, the probability 
that $X_i=1$ is exactly $x_i^{(0)}$ and as the  walk in each coordinate is independent, $X$ has the same distribution on $\{0,1\}^n$ as under randomized rounding.

Now consider iterated rounding, and recall that here the update $y^{(k)}$ at step $k$ must lie in the nullspace of $W^{(k)}$. So a natural first idea to combine this with randomized rounding, is to do a random walk in the null space of $W^{(k)}$ until some variable reaches $0$ or $1$. The slack condition $\text{dim}(W^{(k)}) \leq (1-\delta) n_k$ implies that the nullspace has at least $\delta n_k$ dimensions, which could potentially give ``enough randomness" to the random walk.

It turns out however that doing a standard random walk in the null space of $W^{(k)}$ does not work. The problem is that as the constraints defining $W^{(k)}$
can be completely arbitrary in our setting, the random walk can lead to very high correlations between certain subsets of coordinates causing the $\beta$-concentration property to fail. For example, suppose $\delta=1/2$ and $W^{(0)}$ consists of the $n/2$ constraints $x_1=x_2,  x_2=x_3,\ldots, x_{n/2-1} = x_{n/2}$. Then the random walk will update $x_{n/2},\ldots,x_n$ independently, but for $x_1,\ldots,x_{n/2}$ the updates must satisfy $\Delta x_1 = \ldots = \Delta x_{n/2}$, and hence will be completely correlated. So the linear function $x_1 + \ldots + x_{n/2}$ will have very bad concentration (as all the  variables will simultaneously rise by $-\delta$ or by $+\delta$).

To get around this problem, we design a different random walk that is correlated, but still looks like an almost independent walk  in {\em every} direction. 
More formally, consider a random vector $Y=(Y_1,\ldots,Y_n)$, where $Y_i$ are mean-zero random variables. We say
that $Y$ is $\eta$-almost pairwise independent
if for every $a=(a_1,\ldots,a_n) \in \R^n$,
 \[ \E[ ( a \cdot Y)^2  ]  =  \E [(\sum_i a_i Y_i)^2 ]\leq  \eta  \sum_i a_i^2 \E[Y_i^2] \]
If $Y_1,\ldots,Y_n$ are pairwise independent, note that the above holds as equality  with $\eta=1$, and hence  this can be viewed as a relaxation of pairwise independence.
We show that whenever $\dim(W^{(k)}) \leq (1-\delta) n_k$, there exist $\gamma$-almost pairwise independent random updates $y^{(k)}$ that lie in the null space of $W^{(k)}$ with $\gamma \approx 1/\delta$.
Moreover these updates can be found by solving a semidefinite program (SDP). 

Next, using a variant of Freedman's martingale analysis, we show that applying these almost pairwise independent random  
updates (with small increments) until all the variables reach $0$-$1$, gives an integral solution that satisfies $O(\eta)$-concentration.  

These techniques are motivated by our recent works  
on algorithmic discrepancy \cite{BDG16, BG17}.
While discrepancy is closely related to rounding 
\cite{LSV86, R12}, a key difference in discrepancy is that the error for
rounding a linear system $Ax=b$ depends on the $\ell_2$ norms of the coefficients of the constraints and not
on $b$. E.g.~suppose $x \in [0,1]^n$ satisfies  $x_1 + \ldots + x_n = \log n$, then the sum stays
$O(\log n)$ upon randomized rounding with high probability, while
using discrepancy methods directly gives $\Omega(\sqrt{n})$ error, which would be unacceptably large in this setting. So our results can be 
viewed as using techniques from discrepancy to obtain bounds that are sensitive to $x$. 
Recently, this direction was explored in \cite{BN16}
but their method gave much weaker results and applied to very limited settings.

\section{Technical Preliminaries}
\label{s:prel}
\subsection{Probabilistic tail bounds and Martingales}
The standard Bernstein's probabilistic tail bound for independent random variables  is the following.
\begin{theorem}
\label{thm:berns}(Bernstein's inequality.) Let $Y_1,\ldots,Y_n$ be independent random variables,
with $|Y_i - \E[Y_i]| \leq M$ for all $i \in [n]$.
Let $S = \sum_i Y_i$ and $t>0$. Then, with $\sigma_i^2 = \E[Y_i^2] - \E[Y_i]^2 $ we have,
\[ \Pr[ S - \E[S] \geq t] \leq \exp\left(-\frac{t^2}{2(\sum_i \sigma_i^2 + Mt/3)}\right) \]
\end{theorem}

The lower tail follows by applying the above to $-Y_i$, so we only consider the upper tail.
As we will be interested in bounding  $S =\sum_i a_i X_i$, where the random variables
$X_i$ are $0$-$1$ and the $a_i$ are arbitrary reals (possibly negative), we will use the form given by \eqref{eq:berns}. 

The well-known Chernoff bounds correspond to the special case of \eqref{eq:berns} when $a_i \in [0,1]$ for $i \in [n]$. In particular, setting $t= \eps \mu$ with $\mu  = \E[S] = \sum_i a_i x_i$, $M = 1$ and using that $\sum_i a_i^2 (x_i - x_i^2)  \leq \sum_i a_i^2 x_i \leq \sum_i a_i x_i \leq \mu$ in
\eqref{eq:berns}, we get
\begin{equation}
\label{eq:upper}
 \Pr[X  \geq (1+\epsilon) \mu ] \leq \exp \left( -\epsilon^2 \mu/(2 ( 1+\eps/3)) \right).
\end{equation}

Remark: For $\epsilon>1$, the bound \eqref{eq:upper} can be improved slightly to $ \left(e^\eps/((1+\eps)^{1+\eps})\right)^\mu$
by optimizing the choice of parameters in the proof. In this regime, an analogous version of Theorem \ref{thm:berns} is called Bennett's inequality (\cite{BLM13}), and similar calculations also give such a variant of Theorem \ref{thm:main}. As this is completely standard, we do not discuss this here.

We will use the following Freedman-type martingale inequality.
\begin{lemma}
\label{lem:prel2}
Let $\{Z_k: k=0,1,\ldots,\}$ be a sequence of random variables with $Y_k:=Z_k - Z_{k-1}$, such that $Z_0$ is deterministic and $Y_k\leq 1$ for all $k=1,\ldots$. If  for all  $k=1,2,\ldots,$
\[\E_{k-1}[Y_k] \le  -\alpha \E_{k-1}[Y_k^2]\] 
where $\E_{k-1}[\ \cdot \ ]$ denotes $ \E[\ \cdot \ | Z_1,\dots, Z_{k-1}]$.
Then for all $0< \alpha  < 1$ and $t \geq 0$, it holds that  
\[ \Pr[Z_k  - Z_0 > t] \leq \exp(-\alpha t).\]
\end{lemma}
Before proving Theorem \ref{lem:prel2}, we first need a simple lemma.
\begin{lemma}
\label{lem:prel1}
If $X\le 1$ and $\lambda >0$, 
$ \E[e^{\lambda X}] \le \exp\left(\lambda \E[X] +  (e^{\lambda} -\lambda-1 )\E[X^2]\right)$.
\end{lemma}
\begin{proof}
Let $f(x) = (e^{\lambda x} - \lambda x - 1)/x^2$, where we set $f(0)=\lambda^2/2$. It can be verified that $f(x)$ is increasing for all $x$, which implies that for any $x \le 1$, $e^{\lambda x} \leq 1 + \lambda x  + f(1)x^2 = 1 + \lambda x +   (e^{\lambda} -\lambda-1 )x^2.$ 
Taking expectations and using that $1+x \leq e^x$ for all $x \in \R$ this gives,
\[ \E[e^{\lambda X}]  \leq  1+\lambda\E[ X]+(e^{\lambda} -\lambda-1 )\E [X^2]   \leq \exp(\lambda \E[X] +(e^{ \lambda} -\lambda-1 )\E [X^2]). \qedhere\]
\end{proof}
\begin{proof}(Lemma \ref{lem:prel2})
By Markov's inequality, 
\[ \Pr[Z_k - Z_0 > t ]  = \Pr[\exp(\alpha (Z_k-Z_0)) \geq \exp(\alpha t)] \leq \frac{\E[\exp(\alpha (Z_k-Z_0))] }{\exp(\alpha t)}\]
so it suffices to show that $\E[ \exp(\alpha (Z_k-Z_0)) ]  \leq 1$. As $Z_0$ is deterministic,  this is same as $ \E[\exp(\alpha Z_k)] \leq  \exp(\alpha Z_0)$. Now,
\begin{eqnarray*}
&&\E_{k-1}\left[e^{\alpha Z_k}\right] =  e^{\alpha Z_{k-1}}\E_{k-1}\left[e^{\alpha(Z_k-Z_{k-1})}\right]  =  e^{\alpha Z_{k-1}}\E_{k-1}\left[e^{\alpha Y_k}\right] \\
& \le &  e^{\alpha Z_{k-1}} \exp\left( \alpha \E_{k-1}[Y_k] +(e^{ \alpha} -\alpha-1 )\E_{k-1}\left[Y_k^2\right]\right)    \qquad \textrm{(Lemma \ref{lem:prel1})}\\
 & \le & e^{\alpha Z_{k-1}}\exp\left((e^{ \alpha} -\alpha^2 - \alpha-1 )\E_{k-1}\left[Y_k^2\right]\right)  \leq e^{\alpha Z_{k-1}}   \qquad \textrm{(as $e^\alpha \leq 1 + \alpha + \alpha^2$ for $0 \leq \alpha \leq 1$)} 
\end{eqnarray*}
As this holds for all $k$, using that $\E[\ \cdot \  ] = \E_0[\E_1[ \cdots \E_{k-1}[\ \cdot \  ]]]$ gives the result.
\end{proof}

%
%
%

\subsection{Semidefinite Matrices}
\label{s:psd}
Let $M_n$ denote the class of all symmetric  $n\times n$ matrices with real entries.
For two matrices $A,B \in \mathbb{R}^{n \times n}$,
the trace inner product of $A$ and $B$ is defined as 
$\langle A, B\rangle = \textrm{tr}(A^TB) = \sum_{i=1}^n \sum_{j=1}^n a_{ij}b_{ij}.$
A matrix $U \in M_n$ is positive semidefinite (psd) if all its eigenvalues are non-negative and we note this by $U \succeq 0$.  
Equivalently, $U \succeq 0$ iff   $a^TUa = \langle aa^T, U \rangle \geq 0$ for all $a \in \R^n$.

For, $U \succeq 0$ let $U^{1/2} = \sum_i \lambda_i^{1/2} u_iu_i^T$, where $
U = \sum_i \lambda_i u_iu_i^T$ is the spectral decomposition of $Y$ with orthonormal eigenvectors $u_i$. Then $U^{1/2}$ is psd and $U = V^T V$ for $V= U^{1/2}$.
For $Y,Z \in M_n$, we say that $Y \preceq Z$ if $Z-Y  \succeq 0$.

\subsection{Approximate independence and  sub-isotropic random variables}
Let $Y=(Y_1,\ldots,Y_n)$ be a random vector with $Y_1,\ldots,Y_n$ possibly dependent.
\begin{definition}[$(a,\eta)$ sub-isotropic random vector] For $ a \in (0,1]$ and $ \eta \geq 1$,  We say that $Y$ is $(a,\eta)$ sub-isotropic if it satisfies the following conditions.
\begin{enumerate} 
\item $\E[Y_i]=0$ and $\E[Y_i^2 ] \leq 1$ for all $i \in [n]$, and $\sum_{i=1}^n  \E[Y_i^2] \geq a n$.
\item For all $c = (c_1,\ldots,c_n) \in \R^n$ it holds that 
\begin{equation}
\label{eq:almost-indep}
     \E [( \sum_{i=1}^n c_i Y_i)^2   ] \leq \eta \sum_{i=1}^n c_i^2 \ \E[Y_i]^2.
\end{equation}
\end{enumerate}
\end{definition}
Note that if $Y_1,\ldots,Y_n$ are independent then \eqref{eq:almost-indep} holds with equality for $\eta =1$.

Let $U \in M_n$ be the $n\times n$ covariance matrix of $Y_1,\ldots,Y_n$. That is, $U_{ij} = \E[Y_i Y_j]$. Every covariance matrix is psd as 
$ c^T U c =\E[ (\sum_i c_i Y_i)^2 ] \geq 0$ for all $c \in \R^n$.
Let $\text{diag}(U)$ denote the diagonal $n\times n$ matrix with entries $U_{ii}$, then 
\eqref{eq:almost-indep} can be written as $c^T (\eta \text{ diag}(U) - U ) c \geq 0$ for every $c \in \R^n$, and hence 
equivalently expressed as 
\[  U \preceq \eta  \text{ diag}(U).\]

{\bf Generic construction.} We will use the following generic way to produce $(a,\eta)$ sub-isotropic vectors. Let $U$ be a $n\times n$ PSD matrix satisfying: $U_{ii} \leq 1$, $\text{Tr}(U) \geq an$ and $ U \preceq \eta \text{ diag}(U)$. Let $r \in \R^n$ be  a random vector where each coordinate is independently and uniformly $\pm 1$. Then it is easily verified that 
$Y = U^{1/2}r$ is an $(a,\eta)$ sub-isotropic random vector, as its covariance vector $\E[YY^T]= U^{1/2} \E[rr^T] (U^T)^{1/2} = U$.

{\em Remark:} Typically $r$ is chosen to be a random Gaussian, but random $\pm 1$ will be more preferable for us as it is bounded and this makes some technical arguments easier.   

We will need the following result from \cite{BG17}, about finding sub-isotropic random vectors orthogonal to a subspace.
\begin{theorem}[\cite{BG17}]
\label{thm:bg17}
Let $W \subset \R^n$ be a subspace with dimension $dim(W)=\ell =  (1-\delta) n$. Then for any $a>0 $ and $\eta>1$ satisfying $1/\eta + a \leq \delta$, there is a $n \times n$ PSD matrix $U$ satisfying the following: (i) $\langle  ww^T,U \rangle =0$  for all $w \in W$, (ii)
$U_{ii} \leq 1$ for all $i\in [n]$, (iii) $\text{Tr}(U) \geq an$ and  (iv) $ U \preceq \eta \text{ diag}(U)$.
\end{theorem}
Note that the condition $\langle  ww^T,U \rangle = 0$ is equivalent to $w^T U w = \| U^{1/2} w \|^2=0$, which gives that $w^T U^{1/2}$ is the all-zero vector $\bf{0}$. So, for $Y= U^{1/2} r$,  $w$ is orthogonal to $Y$ as 
\[ w^T Y=  w^T U^{1/2} r = \bf{0} \cdot r = 0.\] 
So this gives a polynomial time algorithm to find a $(a,\eta)$ sub-isotropic random vector $Y \in \R^n$ such that $Y$ is orthogonal to the subspace $W$ with probability $1$.

\subsection{Formal Description of Iterated Rounding}
\label{s:iter-round}
By iterated rounding we refer to any procedure that works as follows.
Let $x$ be the starting fractional solution. We set $x^{(0)}=x$, and round it to a $0$-$1$ solution by
applying a sequence of updates as follows.
Let $x^{(k)}$ denote the solution at the beginning of iteration $k$. We say that variable $i \in [n]$
frozen if $x_i^{(k)}$ is $0$ to $1$, otherwise it is alive. Frozen variables are never updated anymore.  Let $n_k$ denote the number of alive variables.

Based on $x^{(k)}$, the algorithm picks a set of constraints of rank at most $r \leq n_k-1$,
given by the rows of some matrix $W^{(k)}$. It finds an
(arbitrary) non-zero direction $y^{(k)}$ such that $W^{(k)} y^{(k)} =0$ and $y^{(k)}_i=0$ if $i$ is frozen.
The solution is updated as $x^{(k+1)} = x^{(k)} + y^{(k)}$.

In typical applications of iterated rounding, 
$W^{(k)}$ is obtained by dropping one or more rows of $W^{(k-1)}$, and $x^{(k+1)}$ is obtained in a black-box way by solving the LP given by the constraints $W^{(k)} x^{(k+1)} = b^{(k)}$ where $b^{(k)} = W^{(k)} x^{(k)}$, restricted to  the alive variables
 and $x^{(k+1)} \in [0,1]^n$.
As $\text{dim}(W^{(k)}) <n_k $, at least one more variable in $X^{(k+1)}$ reaches
$0$ or $1$ and hence the algorithm terminates in at most $n$ steps.
However, as we will not work with basic feasible solutions, we will view the processing of generating the update $y^{(k)}$ by taking a step in the null space of $W^{(k)}$ as described above.

\section{Rounding Algorithm}

We assume that the problem to be solved has an iterated rounding procedure, as discussed in Section \ref{s:iter-round} that in any iteration $k$ specifies some subspace $W^{(k)}$ with $\dim(W^{(k)}) \leq (1-\delta) n_k$, and the update $y^{(k)}$ must satisfy $W^{(k)} y^{(k)}=0$
We now describe the rounding algorithm. 

{\bf Algorithm.} Initialize $x^{(0)}=x$, where $x$ in the starting fractional solution given as input. 
For each iteration $k=0,1,\ldots,$ repeat the following until all variables reach $0$ or $1$.

{\bf Iteration $k$.} Let $x^{(k)}$ be the current solution and let $A_k \subset [n]$ be the subset of coordinates $i$ for which $x^{(k)}_i \in (0,1)$. Call these alive variables and only the variables in $A_k$ will be updated henceforth. So for ease of notation we assume that $A_k = [n_k]$.
\begin{enumerate}
\item
Apply theorem \ref{thm:bg17}, with $n=n_k$, $W=W^{(k)}$, $a=\delta/10$ and $\eta= 10/(9\delta)$ to find the covariance matrix $U$.
\item  Let $\gamma = 1/(2n^{3/2})$. Let $r_k \in \R^{n_k}$ be a random vector with independent $\pm 1$ entries. Set 
\[x^{(k+1)} = x^{(k)} +  y^{(k)}  \quad \text{ with } \quad y^{(k)} = \gamma_k(U^{1/2} r_k)\] where 
$\gamma_k$ is the largest value  in $(0, \gamma]$, such that both
 $x^{(k)} + y^{(k)}$ and $x^{(k)} - y^{(k)}$ lie in  $[0,1]^{n_k}$.
\end{enumerate}

\section{Algorithm Analysis}
Let $X$ denote the final solution. The property that $\E[X_i]=x_i$ follows directly as the update $y^{(k)}$ at each time step has mean zero in each coordinate. 
As the algorithm always moves in the nullspace of $W^{(k)}$, clearly it will also satisfy the 
iterated rounding guarantee  with probability $1$.

To analyze the running time, we note that whenever $\gamma_k < \gamma$, there is at least $1/2$ probability that some new variable will reach $0$ or $1$ after the update (as the new solution is either $x^{(k)} + y^{(k)}$ or $x^{(k)} - y^{(k)}$ with probability half each). So, in expectation there are at most $2n$ such steps. So we focus on the iterations where $\gamma_k =\gamma$.

Let us define  the energy of $x^{(k)}$ as  $E^{(k)} := \sum_i (x_i^{(k)})^2$. Then after the update,
 $E^{(k)}$ rises is expectation by at least $\gamma^2 n_k a \geq \gamma^2 a$, as,
\begin{eqnarray*}
 \E_k[E^{(k+1)}]  - E^{(k)} & =  & \gamma^2 \E_k [ \sum_i (2 x_i(k) y_i(k) + y_i^2(k))] \\
& = & 
 \gamma^2 \sum_i \E_k[(y_i{(k)})^2] =  \gamma^2 \text{Tr}(U) \geq \gamma^2 a n_k \geq \gamma^2 a.
\end{eqnarray*}
where the second equality uses that $\E_{k}[y_i(k)] =0$, and the last step uses that $n_k \geq 1$.
By a standard argument \cite{B10}, this implies that the algorithm terminates in $O(n/\gamma^2)$ time with constant probability.

{\em Remark:}
One can also make the energy increase  deterministic and get an improved running time of $O( (\log n)/\gamma^2)$. We add an extra constraint that $y^{(k)}$ be orthogonal to $x^{(k)}$. This ensures that term 
$\sum_i (2 x_i(k) y_i(k)$ is always zero and so the energy rises by exactly $\gamma^2 \sum_i (y_i{(k)})^2$.
As this is only one additional linear constraint,  it can increase the dimension of $W^{(k)}$ by at most $1$, and we still have that $\dim(W^{(k)}) \leq (1-\delta/2) n_k$, as long as $n_k = \Omega(1/\delta)$. When $n_k$ becomes smaller, one can revert to the analysis above. 

It remains to show that the rounding satisfies the concentration property, which we do next.
\subsection{Isotropic updates imply concentration}
Let $X_i$ be the final rounded solution, and fix some linear function $S = \sum_i a_i X_i$. 
We will show that 
\[ \Pr[S - \E[S]  \geq t] \leq \exp \left( - \frac{t^2/\beta}{2 (\sum_i a_i^2(x_i-x_i^2) + Mt/3)} \right)  \]
for $\beta = 18 \eta$ which is  $20/\delta$ by  the choice of the parameters in the algorithm.
\begin{proof}By scaling of $a_i$'s and $t$, we can assume that $M=1$.
Let us define the random variable 
\[ Z_k = \sum_i a_i x_i^{(k)} + \lambda \sum_i  a_i^2 x_i^{(k)}(1-x_i^{(k)}),\] 
where  $\lambda \leq 1$ will be optimized later.

Initially, 
$Z_0 = \mu + \lambda v$ where $\mu = \sum_i a_i x_i^{(0)}$ and $v =  \sum_i a_i^2  x_i^{(0)}( 1-x_i^{(0)})$.

As $x^{(k)} = x^{(k-1)} + y^{(k)}$, a simple calculation gives that
\begin{eqnarray}
Y_k & =&  Z_k - Z_{k-1} = \sum_i a_i (x_i^{(k)} - x_i^{(k-1)}) + \sum_i \left(\lambda a_i^2 ( x_i^{(k)}(1- x_i^{(k)}) - x_i^{(k-1)}(1- x_i^{(k-1)}))   \right) \nonumber \\ 
&= & \sum_i a_i y_i^{(k)} + \sum_i \left(\lambda a_i^2 ( y_i^{(k)} (1- 2 x_i^{(k-1)}  -  y_i^{(k)})) \right).
 \label{eq:yk}
\end{eqnarray}
We now show that $Z_k$ satisfies the conditions of Lemma \ref{lem:prel2} for a suitable $\alpha$.

\begin{claim}
\label{bound-y}
For all $k$,  $|Y_k| \leq 1$ and $\|y^{(k)}\|_2 \leq \gamma n = (n^{-1/2})/2$.
\end{claim}
\begin{subproof}
Recall that $y^{(k)} = \gamma U^{1/2} r_k$. 
As  $U_{ii} \leq 1$, the columns of $U^{1/2}$
have length at most $1$. Applying triangle inequality to the columns of $U^{1/2}$, we have 
$\|U^{1/2} r_k\|_2 \leq \|r_k\|_1$ which is at most $n$.
This gives that $\|y\|_2 \leq \gamma n$.

We now show that $|Y_k|\leq 1$. First we note that the second term in \eqref{eq:yk} is at most $\sum_i |a_i y_i^{(k)}|$. This follows as $|a_i|^2 \leq |a_i|$ (as $M=1$), $\lambda \leq 1$ and 
$1-2x_i^{(k-1)} - y_i^{(k)} \in [-1,1]$ (as $1-x_i^{(k-1)} \in [0,1]$  and 
$ x_i^{(k-1)} + y_i^{(k)} = x_i^{(k)} \in [0,1] $).
Applying Cauchy-Schwarz, and using the bound on $\|y^{(k)}\|_2$ above gives
$ |Y_k| \leq 2 \|a\|_2 \|y^{(k)}\|_2  \leq 2 n^{1/2}
\|y^{(k)}\|_2  \leq 2 \gamma n^{3/2} = 1.$
\end{subproof}

\begin{claim}
\label{cool-calc}
$  \E_{k-1}[Y_k] \leq - (\lambda/8 \eta)  \E_{k-1}[Y_k^2] $
\end{claim}
\begin{subproof}
As $\E_{k-1}[y_i^{(k)}]=0$ for all $i$, and as $x_i^{(k-1)}$ is deterministic conditioned on the randomness until $k-1$,  taking expectations $\E_{k-1}[\cdot ]$ in \eqref{eq:yk} gives that
\begin{equation}
    \label{eq:mean}
 \E_{k-1}[Y_k] = -\lambda \sum_i a_i^2 \E_{k-1}[(y_i^{(k)})^2].
\end{equation}
We now upper bound $\E_{k-1}[Y_k^2]$. 
Using $(a+b)^2 \leq 2 a^2 + 2b^2$ twice for the expression in \eqref{eq:yk}, 
\begin{eqnarray}
Y_k^2  & \leq &  2 (\sum_i a_i y_i^{(k)})^2
+ 2   \lambda^2 \left(\sum_i a_i^2  y_i^{(k)} (1- 2  x_i^{(k-1)} - y_i^{(k)}) \right)^2 \nonumber \\
 & \leq &   2 (\sum_i a_i y_i^{(k)})^2 + 4 \lambda^2 \left( \left(\sum_i a_i^2 y^{(k)} (1- 2  x_i^{(k-1)}) \right)^2  + \left(\sum_i a_i^2  (y_i^{(k)})^2 \right)^2 \right)  \label{eq:star}
\end{eqnarray} 
As $|a_i| \leq 1$ and $\sum_i (y_i^{(k)})^2 \leq \gamma n \leq 1/2$ by Claim \ref{bound-y}, the third term can be bounded as 
\begin{equation}
\label{eq:star1}
 \left(\sum_i a_i^2  (y_i^{(k)})^2\right)^2 \leq  (1/2 ) \sum_i a_i^2  (y_i^{(k)})^2.
\end{equation}

Take expectations $\E_{k-1}[\cdot ]$ in \eqref{eq:star}, we upper bound the terms as follows.
As $y^{(k)}$ is $(a,\eta)$ sub-isotropic, by \eqref{eq:almost-indep}, the first term satisfies
\[\E_{k-1}[(\sum_i a_i y_i^{(k)})^2 ] \leq \eta \sum_i a_i^2 
\E_{k-1}[(y_i^{(k)})^2] \]
Similarly, and as $\lambda \leq 1$, the second term satisfies
\[ \E_{k-1}[(\sum_i a_i^2 y_i^{(k)} (1-2x_i^{(k-1)}))^2 ]\leq \eta \sum_i a_i^4 (1-2x_i^{(k-1)})^2  \E_{k-1}
[(y_i^{(k)})^2] \leq \sum_i a_i^2 \E_{k-1}
[(y_i^{(k)})^2],\]
where the last step uses that $|a_i|\leq 1$ and $|1-2x_i^{(k-1)}| \leq 1$.
To bound the third term, we use \eqref{eq:star1}.
Together this gives that,
 \begin{eqnarray*}
 \E_{k-1}[Y_k^2]  \leq   8 \eta \sum_i a_i^2 \E_{k-1}(y_i^{(k)})^2  =  - (8 \eta/\lambda)  \E[Y_k].
\end{eqnarray*}
As $\lambda,\eta>0$, this is the same as $\E_{k-1}[Y_k] \leq  - (\lambda/8 \eta)  \E[Y_k^2]$
\end{subproof}
By Claim \ref{cool-calc}, we can apply Lemma \ref{lem:prel2}  with $\alpha =  \lambda/8\eta$, provided the conditions for Lemma \ref{lem:prel2} are satisfied. Now, $\alpha \leq 1$ is satisfied  as $\lambda \leq 1$ and $\eta \geq 1$, so it remains to show that  $Y_k< 1$. 

By Lemma \ref{lem:prel2}, this gives that $ \Pr[Z_T - Z_0 \geq t] \leq \exp(-\lambda t/8\eta)$, or equivalently
\begin{equation}
    \label{eq:tail1}
    \Pr[ Z_T - \mu - \lambda v \geq  t] \leq \exp(- t \lambda/4\eta).
\end{equation}
Let $\lambda = t/(t + 2 v)$ (note this satisfies our assumption that $\lambda \leq 1$). Then $\lambda v \leq t/2$ and \eqref{eq:tail1} gives $ 
 \Pr[ Z_T - \mu \geq  3t/2] \leq \exp(- t \lambda/8\eta). $
Setting $t'=3t/2$ and the values of $\lambda,\eta$ gives
 \[ \Pr[Z_T  - \mu \geq t' ] \leq \exp\left(-\frac{t'^2/(18 \eta)}{2(v + t'/3)}  \right) \]
which gives the desired result. 
\end{proof}

\section{Applications}
\subsection{Rounding Column-Sparse LP}
\label{s:bf}
Let $x \in [0,1]^n$ be a fractional solution satisfying $Ax = b$,
where $A \in \R^{m \times n}$ is an arbitrary $m \times n$ matrix. Let $t = \max_{j\in [n]} (\sum_{i\in [m]} |a_{ij}|)$ be the maximum $\ell_1$ norm of the columns 
of $A$. Beck and Fiala \cite{BF81} gave a rounding method to find $X \in \{0,1\}^n$ so that the error row $|AX-b|_\infty \leq t$.
 
{\bf Beck-Fiala Rounding.} We first recall the iterated rounding algorithm of \cite{BF81}. Initially $x^{0}=x$.  
Consider some iteration $k$, and let $A_k$ denote the matrix $A$ restricted to the alive coordinates.
Call row $i$ {\em big} if its $\ell_1$-norm in $A_k$ is strictly more than $t$. The number of big rows is strictly less than $n_k$ as each column as norm at most $t$ and thus
the total $\ell_1$ norm of all entries $A_k$ is at most $tn_k$.
So the algorithm sets $W^{(k)}$ to be the big rows of $A_k$, and applies the iterated rounding update.

We now analyze the error. Fix some row $i$. As long as row $i$ is big, its rounding error is $0$ during the update steps. But when it is no longer big no matter how the remaining alive variables are rounded in subsequent iterations, the error incurred can be at most its $\ell_1$-norm, which is at most $t$.

{\bf Introducing Slack.} To apply Theorem \ref{thm:main}, we can easily introduce $\delta$-slack for any $0 \leq \delta < 1$, as follows. In iteration $k$, call a row big if its $\ell_1$ norm exceeds $t/(1-\delta)$,
and by the argument above the number of big rows is strictly less than $n_k(1-\delta)$.
This gives the following result.

\begin{theorem} 
\label{thm:bf} Given a matrix $A$ with maximum  $\ell_1$-norm of any column at most $t$, and any $x\in [0,1]^n$, then for any $0\leq \delta < 1$ the algorithm returns $X\in \{0,1\}^n$ such that $|A(X-x)|_\infty \leq t/(1-\delta)$, and $\E[X_i]=x_i$ and $X$ satisfies $O(1/\delta)$-concentration. 
\end{theorem}

This implies the following useful corollary. 
\begin{corollary} Given a matrix $M$ with some collection of  rows $A$ such that the columns restricted to $A$ have $\ell_1$ norm at most $t$, then
say setting $\delta=1/2$, the rounding ensures at most error $2t$ for rows of $A$, while the error for other rows of $M$ is similar to that as for randomized rounding.
\end{corollary}

\paragraph{Koml\'{o}s Setting.} 
For  a matrix $A$, let $t_2 = \max_{j \in [n]} (\sum_{i\in [m]} a_{ij}^2)^{1/2}$ denote the maximum $\ell_2$ norm of the columns of  $A$.
Note that $t_2 \leq t$ (and it can be much smaller, e.g.~if $A$ is $0$-$1$, $t_2= \sqrt{t}$).

The long-standing Komlos conjecture (together with a connection between hereditary discrepancy and rounding due to \cite{LSV86}) states that any $x \in [0,1]^n$ can be rounded to $X \in \{0,1\}^n$, so that  $|A(X-x)|_\infty = O(t_2)$.
Currently, the best known bound for this problem is $O(t_2 \sqrt{\log m})$ \cite{Bana98,BDG16}.

An argument similar to that for Theorem \ref{thm:bf} gives the following result in this setting.
\begin{theorem} If $A$ has maximum column $\ell_2$-norm  $t_2$, then given any $x\in [0,1]^n$, the algorithm returns $X\in \{0,1\}^n$ such that $|A(X-x)|_\infty \leq t_2 \sqrt{\log m}$, where $X$ also satisfies $O(1)$-concentration. 
\end{theorem}
\begin{proof} We will apply Theorem \ref{thm:main} with $\delta=1/2$. 
During any iteration $k$, call 
row $i$ big if its squared $\ell_2$ norm in $A_k$ exceeds $2(t_2)^2$. As the sum of squared entries of $A_k$ is at most $(t_2)^2 n_k$, the number big rows is at most $n_k/2$ and we set $W^{(k)}$ to $A_k$ restricted to the big rows.

The $O(1)$ concentration follows directly from Theorem \ref{thm:main}. To bound the error for rows of $A$, we argue as follows. Fix a row $i$. Clearly, row $i$ incurs zero error while it is big. Let $k$ be the first iteration when row $i$ is not big. 
Call $j$ large if $|a_{ij}| \geq t_2/\sqrt{\log m}$, and let $L$ denote the set of these coordinates. As $\sum_j a_{ij}^2 \leq 2t$, $|L| \leq 2\log m$. By Cauchy-Schwarz the rounding error due to coordinates in $L$ is at most 
\[ \sum_{j \in L} |a_{ij}| \leq |L|^{1/2} (  \sum_{j \in L} |a_{ij}|^2)^{1/2}  = O(t_2 \sqrt{\log m}).\]
Applying the $O(1)$-concentration property to the  solution $x^{(k)}$, the rounding error due to the entries not in $L$ in row $i$ satisfies
\[ \Pr\left[ \sum_{j \notin L}  a_{ij} (y_j-x^{(k)}_j)  = \Omega ( t_2 \sqrt{ \log m}) \right] 
 \leq   \exp\left(-\frac{(t_2 \sqrt{\log m})^2}{O( \sum_{j \notin L} a_{ij}^2  + Mt_2 \sqrt{\log m}  )}\right) \]
But as $\sum_{j \notin L} a_{ij}^2 \leq 2 t_2^2  $ and $M \leq \sqrt{t_2/\log m}$, this is $\exp(-\Omega(\log m))$. By choosing the constants suitably above, this can be made arbitrarily smaller than $1/m$ and the result follows by a union bound over the rows.
\end{proof}


\subsection{Makespan minimization on unrelated machines}
\label{s:makespan}
In the unrelated machine setting, there are $r$ jobs (we use $n$ for the number of fraction variables) and $m$ machines, and each job $j \in [r]$ has size $p_{ij}$ on a machine $i \in [m]$.
The goal is to assign all jobs to machines to minimize the maximum machine load.

{\bf LP Formulation.} The standard LP relaxation has fractional assignment variables $x_{ij} \in [0,1]$ for $j \in [r]$ and $i\in [m]$. Consider the smallest target makespan $T$ for which the following LP is feasible.
\begin{eqnarray*}
\sum_{j \in [r]} p_{ij} x_{ij}  \leq  T  \qquad & \forall i \in [m]  \qquad \text{(load constraints)}\\
\sum_{i\in [m]} x_{ij}  =  1    \qquad & \forall j \in [r]  \qquad \text{(assignment constraints)}
\end{eqnarray*}
Setting $x_{ij}=0$ if $p_{ij} >T$, which is a valid constraint for the integral solution, we can also assume that 
$p_{\max}:=\max_{ij} p_{ij}$ is at most $T$.
In a classic result, \cite{LST} gave a rounding that gives  makespan $T + p_{\max}$. 
We now  sketch the iterated rounding based proof of this result from \cite{Lau-book}.

{\bf Iterated rounding proof.} As always, we start with $x^{(0)}=x$ and fix the variables that get rounded to $0$ or $1$.
Consider some iteration $k$. Let $n_k$ denote the number of fractional variables, and let
 $R_k$ denote the set of jobs that are still not integrally assigned to some machine. For a machine $i$, define the excess as 
\begin{equation}
\label{eq:excess} e_i:= \sum_{j \in R_k: x^{(k)}_{ij}>0} (1-x^{(k)}_{ij}), 
\end{equation}
and note that 
$e_i$ is the maximum extra load
that $i$ can possibly incur (if all the non-zero variables are rounded to $1$). 
A nice counting argument \cite{Lau-book} shows that if $W^{(k)}$ consists of load constraints for machines with $e_i > p_{\max}$, and assignment constraints for jobs in $R_k$, then $\dim(W^{(k)}) < n_k$.

{\bf Introducing slack.} We now modify the argument of \cite{Lau-book} to have some slack and apply Theorem \ref{thm:main}. This will give the  following result.
\begin{theorem}
\label{thm:lst}
Given any $\delta \in [0,1/2)$, and a fractional solution $x$ to the problem, there is a rounding where the integral solution $X$ increases the load on any machine by $p_{\max}/(1-2\delta)$,  satisfies  $\E[X_{ij}]=x_{ij}$ for all $i,j$ and has $O(1/\delta)$ concentration.
\end{theorem}
\begin{proof}
Without loss of generality, let us assume that $p_{\max}=1$.
Consider some iteration $k$,
and let $n_k$ denote the number of fractional variables $x^{(k)}_{ij} \in (0,1)$, and $R_k$ denote the jobs that are still not integrally assigned. Let $r_k=|R_k|$.
For a machine $i$, we define the excess $e_i$ as in \eqref{eq:excess}. Let $M_k$ denote the set of machines with $e_i > 1/(1-2 \delta)$.  

$W^{(k)}$ will consist of load constraints for machines in $M_k$ and assignment constraints for jobs in $R_k$. More precisely, the update $y^{(k)}_{ij}$ will satisfy:
$ \sum_j p_{ij} y^{(k)}_{ij} = 0$   for all $i \in [M_k]$ and 
$ \sum_i y^{(k)}_{ij}  = 0$, for all $j \in [R_k]$.
We say that $i$ is protected in iteration $k$ if $i \in M_k$. For a protected machine, the fractional load does not change on an update and its the excess can only decrease 
(when $x_{ij}$ reaches $0$ for some $j \in R_k$). 
So all machines in $M_k$ were also protected in previous iterations. When a machine ceases to be protected, the excess ensures that its extra load can be at most $p_{\max}/(1-2\delta)$.

It remains to show that $\dim(W_k) \leq (1-\delta) n_k$.
As each job in $R_k$ contributes at least two fractional variables to $n_k$, we first note that 
\begin{equation}
    \label{eq:rk-nk}2r_k \leq n_k
\end{equation} 
Next, we claim that
\begin{claim}
\label{cl:mknkrk}
$m_k \leq (1-2\delta) (n_k-r_k)$.
\end{claim} 
\begin{subproof}
Clearly $m_k/(1-2\delta) \leq \sum_{i \in M_k}  e_i$, as each  $i\in M_k$ has excess more than $1/(1-2\delta)$.
Next,
\[   \sum_{i \in M_k} e_i =   \sum_{i \in M_k} \sum_{j \in R_k: x^{(k)}_{ij}>0}  (1-x^{(k)}_{ij}) \leq   \sum_{i \in M} \sum_{j \in R_k: x^{(k)}_{ij}>0}  (1-x^{(k)}_{ij})  = n_k -r_k,\]
where the first equality uses the definition of $e_i$ and second uses the definition of $n_k$ and that for each job $j \in R_k$,  $\sum_{i  \in M} x^{(k)}_{ij} =1$.

Together this gives $m_k \leq (1-2\delta)(n_k-r_k)$.
\end{subproof}
Multiplying \eqref{eq:rk-nk} by $\delta$ and adding to the inequality in Claim \ref{cl:mknkrk} gives $  m_k + r_k \leq (1-\delta) n_k$, which implies the result as $\dim(W_k) \leq r_k +m_k$.
\end{proof}

{\em Remarks:} Setting $\delta = 0$ recovers the additive $p_{\max}$ result of \cite{LST}. Theorem \ref{thm:lst} also generalizes directly to $q$ resources, where job $j$ has load vector $p_{ij} = (p_{ij}(1),\ldots,p_{ij}(q))$ on machine $i$, and the goal is to find an
assignment $A$ to minimize $\max_h \max_i (\sum_{j: A(j)=i} p_{ij}(h))$.
A direct modification of the proof above gives a $q p_{\max}/(1-2\delta) $ error, with $O(\delta)$ concentration.
%

\subsection{Minimum cost degree bounded matroid basis}
\label{s:spanning-tree}
Instead of just the degree bounded spanning tree problem, we consider the more general matroid setting as all the arguments apply directly without additional work.

{\bf Minimum cost degree bounded matroid problem (DegMat).} 
The input is a matroid $M$ defined
on elements $V = [n]$ with costs $c:V \rightarrow \R^+$ and $m$ ``degree constraints” specified by $(S_j , b_j)$ for $j\in [m]$,
where
$S_j \subset [n]$ and $b_j \in \mathbb{Z}^+$. The goal is to find a minimum-cost base $I$ in $M$ satisfying the degree bounds, i.e.~$|I \cap S_j | \leq b_j$ for all $j \in [m]$. The matroid $M$
is given implicitly, by an independence oracle (which given a query $I$, returns whether $I$ is an independent set or not).

{\bf Iterated rounding algorithm.}
The natural LP formulation has the variables $x_i \in [0,1]$ for each element $i\in[n]$ and the goal is to minimize the cost $ \sum_i c_i x_i$, subject to the following constraints.
\begin{eqnarray*}
\sum_{i \in S} x_i & \leq & r(S)  \qquad \forall S \subseteq [n]   \qquad  \text{(rank constraints)}\\
\sum_{i \in V} x_i &= & r(V) \qquad \qquad \qquad  \text{(matroid base constraint)} \\
\sum_{i \cap S_j} x_i & \leq  & b_j  \qquad \forall j \in [m]  \qquad \text{(degree constraints)}
 \end{eqnarray*}
Here $r(\cdot)$ is the rank function of $M$.

Given a feasible LP solution with cost $c^*$,
\cite{KLS08,BKN09} gave an iterated rounding algorithm that finds a solution with cost at most $c^*$ and degree violation at most $q-1$. 
Here $q := \max_{i \in [m]} (\sum_{j \in [m]} |S_j \cap \{i\}|)$ is the maximum number of sets $S_j$ that any element $i$ lies in. Note that $q=2$ for
degree bounded spanning tree, as the elements are edges and $S_j$ consist of edges incident to a vertex.

We briefly sketch their argument. We start with $x^{(0)}=x$ and apply iterated rounding as follows.
Consider some iteration $k$. Let $A_k$ denote the set of fractional variables and let $n_k=|A_k|$. For a set $S_j$, define the excess as 
\begin{equation}
\label{eq:excess2} e_j:= \sum_{i \in A_k \cap S_j} (1-x^{(k)}_{i}), 
\end{equation}
the maximum degree violation for $S_j$ even if all current fractional variables are rounded to $1$.

Let $D_k$ be the set of indices $j$ of degree constraints with excess $e_j \geq q$.
The algorithm chooses $W^{(k)}$ to be the degree constraints in $D_k$ (call these protected constraints) and some basis for the tight matroid rank constraints.
An elegant counting argument then shows that $\dim(W_k) \leq n_k-1$. The correctness follows as if a degree constraint is no longer protected, then its excess is strictly below $q$, and by integrality of $b_j$ and the final rounded solution, the degree violation can be at most $q-1$.

{\bf Introducing slack.} We will modify the argument above in a direct way to introduce some slack, and then apply Theorem \ref{thm:main}.
This will imply the following.
\begin{theorem} For any $0 < \delta < 1$, there is  an algorithm that produces a basis with degree violation strictly less than $q/(1-2\delta)$, and satisfies $\Omega(\delta)$-concentration.
\end{theorem}
Setting $\delta < 1/6 - \epsilon$, and noting that $2/(1-\delta) < 3$, we get  the following.
\begin{corollary}
For minimum cost degree bounded spanning tree, given a fractional solution $x$ there is an algorithm to find a tree satisfying $+2$ degree violation and $O(1)$-concentration.
\end{corollary}

We now describe the argument. Consider iteration $k$. Let $A_k$ be the set of fractional variables and $n_k=|A_k|$. We need to specify how the choose $W^{(k)}$ and show that $\dim(W^{(k)}) \leq (1-\delta)n_k$.
Let $D_k$ denote the set of indices $j$ of degree constraints with excess $e_j \geq  q(1-2\delta)$.
Let $\cal{F}$ denote the family of the tight matroid constraints that holds with equality, i.e.~$\sum_{i \in S \cap A_k} x_i = r_k(S)$, where $r_k$ in the rank function of the matroid $M_k$ obtained from $M$ by deleting elements with $x_i=0$ and contracting those with $x_i=1$.
It is well-known that there is equivalent chain family $\mathcal{C} = \{C_1,\ldots,C_{\ell}\}$, with  $C_1 \subset C_2 \subset \cdots \subset  C_{\ell}$ of tight sets,
such that rank constraint of every $S \in \mathcal{F}$
lies in the linear span of the constraints for sets in $\mathcal{C}$.  Let $c_k = |\mathcal{C}|$ and $d_k=|D_k|$. We set $W^{(k)}$ to be degree constraints in $D_k$ and rank constraints in $\cal{C}$.

\begin{claim}
$\dim(W^{(k)}) \leq (1-\delta) n_k.$
\end{claim}
\begin{proof}
It suffices to show that $c_k + d_k \leq (1-\delta) n_k$.
As each $x_i$ is fractional and as $r_k(C)$ are integral, it follows that any two sets in chain family differ by at least two elements, i.e.~$|C_{i+1}\setminus C_i | \geq 2$.
 This implies that  $c_k \leq n_k/2$.  We also note that $r_k(C_1) < r_k(C_2) \cdots <$ and in particular the rank $r(C_{c_k})$ of the largest set in $\cal{C}$ is at least $c_k$. This gives that $ \sum_{i \in A_k} x_i \geq c_k$.
 
 Next, as $e_j \geq q/(1-2\delta)$ for each $j \in D_k$, we have $q d_k   \leq  (1-2 \delta)  \sum_{j \in D_k}  e_j$.
Moreover, by definition of $e_j$
\[ \sum_{j \in D_k}  e_j  =   \sum_{j \in D_k} \sum_{i \in A_k \cap S_j}  (1-x_i)   = \sum_{i \in A_k}  q_i(1-x_i) \]
where  $q_i = |\{j: i\in S_j, j \in D_k \}$ is the number of tight degree constraints in $D_k$ that contain element $i$.
As $q_i \leq q$, the above is at most $q \sum_{i \in A_k} (1-x_i)  \leq q n_k - q c_k$,
where we use that  $\sum_{i \in A_k} x_i \geq c_k$, and $\sum_{i \in A_k} 1 = |A_k|=n_k$. 

Together this gives, $d_k \leq  (1-2\delta)  (n_k-c_k)$, and adding $2\delta$ times the inequality $c_k \leq n_k/2$  to this gives that $d_k +  c_k \leq  (1-\delta) n_k$ as claimed.
\end{proof}

{\em Remark}: As the underlying LP has exponential size and implicit, some care is needed on how to maintain the chain family and compute the step size of the walk. These issues are discussed in \cite{BN16}.

\subsection{Multi-budgeted Matching}
\label{s:matching}
Consider a bipartite graph with $n$ vertices on each side, and let $x_e$ be a solution to the perfect matching polytope defined by the constraints $\sum_{e \in \delta(v)} x_{e} =1$ for all $v$.
If the support of $x$ is a cycle, then the tight constraints have rank exactly $2n-1$ and there is only one way to write $x$ as a convex combination of perfect matchings: consisting of all odd edges or all even edges in the cycle. In particular, this also shows that no concentration is possible without relaxing the perfect matching requirement.

We sketch an alternate proof of the following result of \cite{CVZ11}.
\begin{theorem}
\label{thm:cvz}
Given a fractional perfect bipartite matching, there is a rounding procedure that given any $\delta>0$, outputs a random matching where each vertex is matched  with probability at least $1-\delta$, each edge occurs with probability $(1-\delta) x_e$, and satisfies $O(1/\delta)$-concentration.
\end{theorem}
 We begin with a simple graph theoretic lemma.
\begin{lemma}
\label{lem:path2}
Let $G=(V,E)$ be connected graph such that $G$ is not a cycle, and has 
no path $P = v_1,\ldots,v_{t}$ of $t$ vertices with each $v_i$ of degree exactly two.
Then for $t\geq 2$,
 $|E| \geq (1+(1/4t)) (d_2 + d_{\geq 3}),$ where $d_2$ (resp.~$d_{\geq 3}$) is the number of degree $2$ (resp. $ \geq 3$) vertices in $G$.
\end{lemma}
\begin{proof}
Consider the decomposition of the edges of $G$ into  paths $P = u,v_1,\ldots,v_\ell,w$ where each $v_i$, $i\in \ell$ has degree two, and $u$ and $w$ have degree other than $2$ (we allow $u=w$ or $\ell=0$). Such a decomposition can be obtained by repeating the following. Start from some vertex of degree other than $2$ (which must exist as $G$ is not a cycle) and visiting previously unexplored edges until some vertex of degree other than $2$ is reached.

For each such path $P$, put a charge of $\ell/2$ on the edges $e=(u,v_1)$ and $e'=(v_{\ell},w)$. 
As $P$ has $\ell$ degree $2$ vertices, the total charge put is $d_2$.
As $e$ and $e'$ are contained in exactly one path $P$, the charge on each edge $e$ is at most $t/2$ (if $e=e'$, then the edge can get contribution twice, but this only happens if $\ell=0$). 
So $d_2 \leq (t/2) (d_1 + \sum_{j \geq 3} j d_j).$
As $2|E| = d_1 + 2d_2 + \sum_{j \geq 3} j d_j$,  this gives 
$|E| \geq d_2 (1+1/t)$. Multiplying this by $1/4$ and adding it to $3/4$ times the trivial bound $|E| \geq d_2 + (3/2) d_{\geq 3}$ gives the result.
\end{proof}
\begin{proof}{(\em Theorem \ref{thm:cvz})}
By Theorem \ref{thm:main}, to get $O(\delta)$-concentration it suffices to show how to create $\delta$-slack. We will show how to maintain while ensuring that any edge or vertex constraint is dropped over all the iterations with probability at most $\delta$. 

Consider some iteration $k$, and let $G_k$ be the graph supported on fractional edges of $x^{(k)}$. 
We assume that $G_k$ is connected, else each component can be considered separately.
Then $n_k =|E(G_k)|$ and $\dim(W_k) \leq d_2 + d_{\geq 3}$ as there are no constraints on degree $1$ vertices. 
To create slack, we will simply drop some small fraction of edges
whenever there is a long path of degree two vertices in $G_k$, and apply Lemma \ref{lem:path2}.
Let $t=\delta/8$. 

For a cycle of length $\ell$, the tight constraints have rank $\ell-1$, as they have the form $y_e - y_{e'}=0$, where $e$ and $e'$ are the edges adjacent to a vertex.
So if $G_k$ is a cycle of length $\ell \leq 8t$, we already have $\delta$-slack.
If $G$ is a cycle of length $\ell \geq 4t$, we pick a random edge and drop this and every consecutive $4t$-th edge. An edge is deleted with probability at most  $1/2t = \delta/2$, and the edges in $G_k$ will not be considered again for dropping after this step as the support will only consist of paths of length at most $4t$.

If $G_k$ is not a cycle and has some path $P$ of degree two vertices of length more than $4t$, we apply the following deletion step. Pick a random offset in $[1,4t]$ and drop that edge and every subsequent $4t$-th edge on $P$. Clearly, an edge is dropped with probability at most $1/(4t)$ in this step. We now show that the edges of $P$ can be considered at most once for deletion in future steps. After the dropping step above, 
$G_k$ breaks into components that are paths of length at most $4t$, and two or fewer other components containing the initial prefix or suffix of $P$ of length at most $4t$. In later iterations this prefix/suffix of $P$ could become a part of a degree two path of length more than $4t$, but after that deletion step these edges will for a component of size at most $4t$ and not considered for deletion anymore.
\end{proof}

\section{Concluding Remarks}
We gave a general approach to obtain concentration properties whenever iterated rounding can be applied with some slack. We also described some applications where iterated rounding can be applied with slack without losing much in the approximation guarantee. While it is quite easy to create slack in most applications of iterated rounding, one example where it not clear to us how to do this is the survivable network design problem for 
which Jain \cite{Jain01} gave a breakthrough $2$-approximation
using iterated rounding (see \cite{NRS10} for a simpler proof).
It would be interesting to see if $O(1)$-concentration can also be achieved here while maintaining an $O(1)$ approximation.

\section*{Acknowledgments}
We thank Shachar Lovett for finding an error in a previous version of the paper.

{ 
\bibliographystyle{plain}
{\small \bibliography{mybib}}
}

\section{Appendix}
\subsection{Tight example for Theorem \ref{thm:main}}
\label{a:tight}
The following simple example shows that the dependence in Theorem \ref{thm:main} cannot improved beyond constant factors.

Let $\delta=1/t$ for some integer $t$.
There are $n$ variables $x_1,\ldots,x_n$, partitioned into $n/t$ blocks $B_1,\ldots,B_{n/t}$ where each block $B_i$ has $t$ variables $ x_{(i-1)t+1}, \ldots, x_{it}$.
For each block we have
$t-1$ constraints given by $x_{(i-1)t+1} =  x_{(i-1)t+2} =  \ldots = x_{it}$, and hence there are $(t-1)(n/t) = n(1-\delta)$ constraints in total.
Given any starting feasible solution, as the algorithm proceeds the variables within a block evolve identically (in iteration $k$, set $W^{(k)}$ consists of the constraints for blocks whose variables are not yet fixed to $0$ or $1$). 
As all the variables in a block  are eventually to the same value, it is easily seen that the linear function $S=x_1+\cdots+x_n$ will only be $1/\delta$-concentrated, as opposed to $1$-concentrated under randomized rounding that rounds all the $n$ variables independently.

\subsection{Thin Spanning Trees}
\label{a:thintree}
Given an undirected graph $G=(V,E)$, a spanning tree $T$ is called $\alpha$-thin, if 
for each subset  $S \subset V$, the number of edges in $T$ crossing $S$ is at most $\alpha$ times that in $G$, i.e.~$\delta_T(S) \leq \alpha \delta_G(S)$.
The celebrated thin tree conjecture states that any $k$ edge-connected graph $G$ has an
$O(1/k)$-thin tree. This conjecture has received a lot of attention recently, due to its connection to the asymmetric TSP problem \cite{AGMOS10}. Despite the recent breakthrough on ATSP \cite{STV18}, the thin-tree conjecture remains open.

Since any $k$-edge connected graph has a packing of $k/2$ edge disjoint spanning trees, 
finding a $\beta/k$-thin tree in a $k$-edge connected graph is equivalent to the following. Given a fractional point $x$ in the spanning tree polytope, find a tree $X$ such that $\delta_X(S) = \beta \delta_x(S)$, where $\delta_x(S)$ is the fraction of edges across $S$ in $x$.
A cut-counting argument from \cite{AGMOS10} shows that any spanning tree $X$ that satisfies concentration satisfies $\beta=O(\log n/\log \log n)$, and by now several negative dependence based methods such as maximum-entropy sampling, swap-rounding, sampling from determinantal measures and Brownian walks \cite{AGMOS10,CVZ10,HO14,PSV17}, are known for sampling such a random tree.

For any fractional spanning tree $x$, the result of \cite{SL07} gives an integral spanning tree with degree of any vertex at most $1$ more than the ceiling of the fractional degree of that vertex. As any fractional spanning tree has vertex degree at most $1$, this implies that the tree produced satisfies $\beta=2$ for vertex cuts.
\end{document}